\documentclass[a4paper,UKenglish,cleveref, autoref, thm-restate]{lipics-v2021}

\usepackage{amsmath,amssymb,amsthm,mathtools,comment}
\usepackage{enumerate}
\usepackage{fontawesome}
\usepackage{float}
\usepackage{graphicx}
\usepackage{todonotes}
\usepackage{complexity}
\graphicspath{ {figures/} }
\usepackage{tikz}
\usetikzlibrary{automata,positioning, decorations.pathreplacing, arrows, decorations.markings}

\def\vecsign{\mathchar"017E}
\def\dvecsign{\smash{\stackon[-1.95pt]{\vecsign}{\rotatebox{180}{$\vecsign$}}}}
\def\dvec#1{\def\useanchorwidth{T}\stackon[-4.2pt]{#1}{\,\dvecsign}}
\usepackage{stackengine}
\stackMath

\theoremstyle{definition}

\newtheorem*{lemma*}{Lemma}
\newtheorem*{theorem*}{Theorem}

\title{Reachability and Matching in Single Crossing Minor Free Graphs} 

\author{Samir Datta}{Chennai Mathematical Institute, Chennai, India}{sdatta@cmi.ac.in}{}{Partially funded by a grant from Infosys foundation and
SERB-MATRICS grant MTR/2017/000480}
\author{Chetan Gupta}{Aalto University, Finland}{chgpt.09@gmail.com}{}{Supported by Academy of Finland, Grant 321901}
\author{Rahul Jain}{Fernuniversit\"at in Hagen, Germany}{rahul.jain@fernuni-hagen.de}{}{}
\author{Anish Mukherjee}{Institute of Informatics, University of Warsaw, Poland}{anish@mimuw.edu.pl}{}{Supported by the ERC CoG grant TUgbOAT no 772346}
\author{Vimal Raj Sharma}{Indian Institute of Technology, Kanpur, India}{vimalraj@cse.iitk.ac.in}{}{Ministry of Electronics and IT, India, VVS PhD program}
\author{Raghunath Tewari}{Indian Institute of Technology, Kanpur, India}{rtewari@cse.iitk.ac.in}{}{Young Faculty Research Fellowship, Ministry of Electronics and IT, India}

\authorrunning{S. Datta, C. Gupta, R. Jain, A. Mukherjee, V.R. Sharma, R. Tewari} 

\Copyright{Samir Datta, Chetan Gupta, Rahul Jain, Anish Mukherjee, Vimal Raj Sharma and Raghunath Tewari} 



\category{} 

\newclass{\ReachUL}{ReachUL}
\newclass{\coUL}{coUL}

\newcommand{\wund}{w^{\textrm{{\tiny{und}}}}}

\newcommand{\pty}{\textit{p-type }}
\newcommand{\cty}{\textit{c-type }}

\newlang{\NZCL}{NonZeroCircL}
\newlang{\NZCNC}{NonZeroCircNC}

\newclass{\Log}{L}
\newclass{\ACz}{AC^0}
\newclass{\TCz}{TC^0}
\newclass{\ACo}{AC^1}
\newclass{\ACzt}{AC^0[\oplus]}
\newclass{\FOar}{FO(\le,+,\times)}
\newclass{\FOpar}{FO[\oplus](\le,+,\times)}
\newclass{\DynACz}{DynAC^0}
\newclass{\DynTCz}{DynTC^0}
\newclass{\DynACzt}{DynAC^0[\oplus]}
\renewclass{\DynFO}{DynFO}
\newclass{\DynFOar}{DynFO(\le,+,\times)}
\newclass{\DynFOp}{DynFO[\oplus]}
\newclass{\DynFOpar}{DynFO[\oplus](\le,+,\times)}

\newlang{\PM}{PM}
\newlang{\BPM}{BPM}
\newlang{\PMD}{PMDecision}
\newlang{\BPMD}{BPMDecision}
\newlang{\PMS}{PMSearch}
\newlang{\BPMS}{BPMSearch}
\newlang{\BMWPMS}{MinWtBPMSearch}
\newlang{\MCM}{MM}
\newlang{\BMCM}{BMM}
\newlang{\BMCMD}{BMMDecision}
\newlang{\BMCMS}{BMMSearch}
\newlang{\MCMSz}{MMSize}
\newlang{\BMCMSz}{BMMSize}
\newlang{\MWMCM}{MinWtMM}
\newlang{\BMWMCM}{MinWtBMM}
\newlang{\BMWMCMS}{MinWtBMMSearch}
\newlang{\Reach}{Reach}
\newlang{\Dist}{Distance}
\newlang{\Rank}{Rank}

\nolinenumbers 

\hideLIPIcs  

\EventEditors{John Q. Open and Joan R. Access}
\EventNoEds{2}
\EventLongTitle{42nd Conference on Very Important Topics (CVIT 2016)}
\EventShortTitle{CVIT 2016}
\EventAcronym{CVIT}
\EventYear{2016}
\EventDate{December 24--27, 2016}
\EventLocation{Little Whinging, United Kingdom}
\EventLogo{}
\SeriesVolume{42}
\ArticleNo{23}

\begin{document}

\maketitle

\begin{abstract}
We show that for each single crossing graph $H$, a polynomially bounded weight function for all $H$-minor free graphs $G$ can be constructed in logspace such that it gives nonzero weights to all the cycles in $G$. This class of graphs subsumes almost all classes of graphs for which such a weight function is known to be constructed in logspace. As a consequence, we obtain that for the class of $H$-minor free graphs where $H$ is a single crossing graph, reachability can be solved in $\UL$, and bipartite maximum matching can be solved in $\SPL$, which are small subclasses of the parallel complexity class $\NC$.  In the restrictive case of bipartite graphs,  our maximum matching result improves upon the recent result of Eppstein and Vazirani~\cite{EV21}, where they show an $\NC$ bound for constructing perfect matching in general single crossing minor free graphs.
\end{abstract}
\section{Introduction}
Directed graph reachability and perfect matching are two fundamental problems in computer science. The history of the two problems has been inextricably linked together from the inception of computer science (and before!) \cite{FordF56}. The problems and their variants, such as shortest path \cite{Dijkstra59} and maximum matching \cite{Edmonds65} have classically been studied in the sequential model of computation. Since the 1980s, considerable efforts have been spent trying to find parallel algorithms for matching problems spurred on by the connection to reachability which is, of course, parallelizable.
The effort succeeded only in part with the discovery of randomized parallel algorithms \cite{KUW85,MulmuleyVV87}. While we know that the reachability problem is complete for the complexity class $\NL$, precise characterization has proved to be elusive for matching problems. The 1990s saw attempts in this
direction when surprisingly ``small'' upper bounds were proved \cite{ARZ99}
for the perfect matching problem, although in the non-uniform setting.
At roughly the same time, parallel algorithms for various
versions of the matching problem for restricted graph classes like planar
\cite{MN95} and bounded genus \cite{MV00} graphs were
discovered. The last two decades have seen efforts towards pinning down the
exact
parallel complexity of reachability and matching related problems in restricted
graph classes \cite{BTV09,KV10,DKR10,DKTV11,AGGT16,KT16,GST19,GST20}.
Most of these papers are based on the method of constructing \emph{nonzero circulations}.

The circulation of a simple cycle is the sum of its edge-weights in a
fixed orientation (see Section~\ref{sec:prelims} for the definition)
and
we wish to assign polynomially bounded weights to the edges of a graph, such
that every simple cycle has a nonzero circulation.
Assigning such weights \emph{isolates} a reachability witness or a matching witness in the graph \cite{TV12}. Constructing
polynomially bounded isolating weight function in parallel for general graphs has been
elusive so far.
The last five years
have seen rapid progress in the realm of matching problems, starting with
\cite{FGT} which showed that the method of nonzero circulations could
be extended from topologically restricted (bipartite) graphs to general
(bipartite) graphs. A subsequent result extended this to all graphs
\cite{ST17}. More recently, the endeavour to parallelize planar
perfect matching has borne fruit \cite{Sankowski18,AV20} and has been followed up by further exciting work \cite{AV2}. 

We know that polynomially bounded weight functions that give nonzero circulation to every cycle can be constructed in logspace for planar graphs, bounded genus graphs and bounded treewidth graphs \cite{BTV09,DKTV11,DKMTVZ20} . Planar graphs are both $K_{3,3}$-free and $K_5$-free graphs. Such a weight function is also known to be constructable in logspace for $K_{3,3}$-free graphs and $K_5$-free graphs, individually \cite{AGGT16}. A natural question arises if we can construct such a weight function for $H$-minor-free graphs for any arbitrary graph $H$. A major hurdle in this direction is the absence of a space-efficient (Logspace) or parallel algorithm (\NC) for finding a structural decomposition of $H$-minor free graphs. However, such a decomposition is known when $H$ is a single crossing graph. This induces us to solve the problem for single crossing minor-free (SCM-free) graphs. An SCM-free graph can be decomposed into planar and bounded treewidth graphs. Moreover, $K_{3,3}$ and $K_5$ are single crossing graphs. Hence our result can also be seen as a generalization of the previous results on these classes. There have also been important follow-up works on parallel algorithms for SCM-free graphs \cite{EV21}. SCM-free graphs have been studied in several algorithmic works (for example \cite{STW16,CE13,DHNRT04}).

\subsection{Our Result}
In this paper, we show that results for previously studied graph classes
(planar, constant tree-width and $H$-minor free for $H \in \{K_{3,3},K_5\}$)
can be extended and unified to yield similar results for SCM-free graphs.
\begin{theorem}
\label{thm:main}
There is a logspace algorithm for computing polynomially-bounded, skew-symmetric nonzero circulation weight function in SCM-free graphs.
\end{theorem}

An efficient solution to the circulation problem for a class of graphs yields better complexity bounds for determining reachability in the directed version of that class and constructing minimum weight maximum-matching in the bipartite version of that class. Theorem~\ref{thm:main}
with the results of \cite{DKKM18,RA00}, yields the following:
\begin{corollary}\label{cor:stat}
For SCM-free graphs, reachability is in $\UL\cap\coUL$ and minimum weight bipartite maximum matching is in $\SPL$.
\end{corollary}

Also using the result of \cite{TW10}, we obtain that the \textit{Shortest path} problem in SCM-free graphs can be solved in $\UL\cap \coUL$.  
\subparagraph*{Overview of Our Techniques and Comparison With Previous Results:} We know that for planar graphs and constant treewidth graphs nonzero circulation weights can be constructed in logspace \cite{BTV09,DKMTVZ20}. We combine these weight functions using the techniques from Arora et al. \cite{AGGT16}, Datta et al. \cite{DKKM18} and, Datta et al. \cite{DKMTVZ20} together with some modifications to obtain the desired weight function. In \cite{AGGT16}, the authors decompose the given input graph $G$ ($K_{3,3}$-free or $K_5$-free) and obtain a component tree that contains planar and constant size components. They modify the components of the component tree so that they satisfy few properties which they use for constructing nonzero circulation weights (these properties are mentioned at the beginning of Section \ref{sec:wtfn}). The new graph represented by these modified components preserves the perfect matchings of $G$. Then, they construct a \emph{working-tree} of height $O(\log n)$ corresponding to this component tree and use it to assign nonzero circulation weights to the edges of this new graph. The value of the weights assigned to the edges of the new graph is exponential in the height of the working tree.

While $K_{3,3}$-free and $K_5$-free graphs can be decomposed into planar and constant size components, an SCM-free graph can be decomposed into planar and constant treewidth components. Thus the component tree of the SCM-free graph would have several non-planar constant treewidth components. While we can construct a working tree of height $O(\log n)$, this tree would contain constant-treewidth components and hence make it difficult to find nonzero circulation weights. A na\"ive idea would be to replace each constant treewidth component with its tree decomposition in the working tree. However, the resultant tree would have the height $O(\log^2 n)$. Thus the weight function obtained in this way is of $O(\log^2 n)$-bit. We circumvent this problem as follows: we obtain a component tree $T$ of the given SCM-free graph $G$ and modify its components to satisfy the same property as \cite{AGGT16} (however, we use different gadgets for modification). Now we replace each bounded treewidth component with its tree decomposition in $T$. Using this new component tree, say $T'$, we define another graph $G'$. We use the technique from \cite{DKKM18} to show that if we can construct the nonzero circulation for $G'$, then we can \textit{pull back} nonzero circulation for $G$. Few points to note here: (i) pull back technique works because of the new gadget that we use to modify the components in $T$, (ii) since ultimately we can obtain nonzero circulation for $G$, it allows us to compute maximum matching in $G$ in $\SPL$, which is not the case in \cite{AGGT16}.

\subsection{Organization of the Paper}
After introducing the definitions and preliminaries in Section~\ref{sec:prelims}, in Section~\ref{sec:wtfn} we discuss the weight function that achieves non-zero circulation in single-crossing minor free graphs and its application to maximum matching in Section~\ref{sec:staticmatch}. Finally, we conclude with Section~\ref{sec:concl}.

\section{Preliminaries and Notations}
\label{sec:prelims}
\subparagraph*{Tree decomposition:} Tree decomposition is a well-studied concept in graph theory. Tree decomposition of a graph, in some sense, reveals the information of how much tree-like the graph is. We use the following definition of tree decomposition.

\begin{definition}
\label{def:treedec}
Let $G(V,E)$ be a graph and $\tilde{T}$ be a tree, where nodes of the $\tilde{T}$ are $\{B_1, \ldots ,B_k \mid B_i \subseteq V\}$ (called bags). $T$ is called a tree decomposition of $G$ if the following three properties are satisfied:
\begin{itemize}

\item $B_1 \cup \ldots \cup B_k = V$,
\item for every edge $(u,v) \in E$, there exists a bag $B_i$ which contains both the vertices $u$ and $v$,
\item for a vertex $v \in V$, the bags which contain the vertex $v$ form a connected component in $\tilde{T}$.
\end{itemize}

\end{definition}

The width of a tree decomposition is defined as one less than the size of the largest bag. The treewidth of a graph $G$ is the minimum width among all possible tree decompositions of $G$. Given a constant treewidth graph $G$, we can find its tree decomposition $\tilde{T}$ in logspace such that $\tilde{T}$ has a constant width \cite{EJT10}.

\begin{lemma}[\cite{EJT10}]
\label{lem:boundtw}
For every constant $c$, there is a logspace algorithm that takes a graph as input and outputs its tree decomposition of treewidth at most $c$, if such a decomposition exists.
\end{lemma}

\begin{definition}
\label{def:cls}
Let $G_1$ and $G_2$ be two graphs containing cliques of equal size. Then the \emph{clique-sum} of $G_1$ and $G_2$ is formed from their disjoint union by identifying pairs of vertices in these two cliques to form a single shared clique, and then possibly deleting some of the clique edges.
\end{definition}
For a constant $k$, a $k$-clique-sum is a clique-sum in which both cliques have at most $k$ vertices. One may also form clique-sums of more than two graphs by repeated application of the two-graph clique-sum operation. For a constant $W$, we use the notation $ \langle \mathcal{G}_{P, W}\rangle_k$ to denote the class of graphs that can be obtained by taking repetitive $k$-clique-sum of planar graphs and graphs of treewidth at most $W$. In this paper, we construct a polynomially bounded skew-symmetric weight function that gives nonzero circulation to all the cycles in a graph $G \in \langle \mathcal{G}_{P, W}\rangle_3$. Note that if a weight function gives nonzero circulations to all the cycles in the biconnected components of $G$, it will give nonzero circulation to all the cycles in $G$ because no simple cycle can be a part of two different biconnected components of $G$. We can find all the biconnected components of $G$ in logspace by finding all the articulation points. Therefore, without loss of generality, assume that $G$ is biconnected.

The \emph{crossing number} of a graph $G$ is the lowest number of edge crossings of a plane drawing of $G$. A \emph{single-crossing} graph is a graph whose crossing number is at most 1. SCM-free graphs are graphs that do not contain $H$ as a minor, where $H$ is a fixed single crossing graph. Robertson and Seymour have given the following characterization of SCM-free graphs.

\begin{theorem}[\cite{RS91}]
\label{thm:rs}
For any single-crossing graph $H$, there is an integer
$c_H \geq 4$ (depending only on $H$) such that every graph with no minor isomorphic to $H$ can
be obtained as $3$-clique-sum of planar graphs and graphs of treewidth at most $c_H$.
\end{theorem}

\subparagraph*{Component Tree:} In order to construct the desired weight function for a graph $G \in \langle \mathcal{G}_{P, W}\rangle_3$, we decompose $G$ into smaller graphs and obtain a component tree of $G$ defined as follows: we first find $3$-connected and $4$-connected components of $G$ such that each of these components is either planar or of constant treewidth. We know that these components can be obtained in logspace \cite{TW09}. Since $G$ can be formed by taking repetitive $3$-clique-sum of these components, the set of vertices involved in a clique-sum is called a separating set. Using these components and separating sets, we define a component tree of $G$. A component tree $T$ of $G$ is a tree such that each node of $T$ contains a $3$-connected or $4$-connected component of $G$, i.e., each node contains either a planar or constant treewidth subgraph of $G$. There is an edge between two nodes of $T$ if the corresponding components are involved in a clique-sum operation. If two nodes are involved in a clique-sum operation, then copies of all the vertices of the clique are present in both components. It is easy to see that $T$ will always be a tree. Within a component, there are two types of edges present, \textit{real} and \textit{virtual edges}. Real edges are those edges that are present in $G$. Let $\{a,b,c\}$(or $\{a,b\}$) be a separating triplet(or pair) shared by two nodes of $T$, then there is a clique $\{a,b,c\}$ (or $\{a,b\}$) of virtual edges present in both the components. Suppose there is an edge present in $G$ between any pair of vertices of a separating set. In that case, there is a real edge present between that pair of vertices parallel to the virtual edge, in exactly one of the components which share that separating set.

\subparagraph*{Weight function and circulation:} Let $G(V, E)$ be an undirected graph with vertex set $V$ and edge set $E$. By $\dvec{E}$, we denote the set of bidirected edges corresponding to $E$. Similarly, by $G(V, \dvec{E})$, we denote the graph corresponding to $G(V, E)$ where each of its edges is replaced by a corresponding bidirected edge. A weight function $w : \dvec{E} \rightarrow \mathbb{Z}$ is called skew-symmetric if for all $e\in \dvec{E}$, $w(e) = -w(e^r)$ (where $e^r$ represent the edge with its direction reversed). We know that if $w$ gives nonzero circulation to every cycle that consists of edges of $\dvec{E}$ then it isolates a directed path between each pair of vertices in $G(V,\dvec{E})$. Also, if $G$ is a bipartite graph, then the weight function $w$ can be used to construct a weight function $\wund : E \rightarrow \mathbb{Z}$ that isolates a perfect matching in $G$ \cite{TV12}.

A convention is to represent by $\left<w_1,\ldots,w_k\right>$ the weight function that on edge $e$ takes the weight $\sum_{i=1}^k{w_i(e)B^{k-i}}$ where $w_1,\ldots,w_k$ are weight functions such that $\max_{i=1}^k{(nw_i(e))} \leq B$.

\subparagraph*{Complexity Classes:} The complexity classes \Log\ and \NL\ are the classes of languages accepted by deterministic and non-deterministic logspace Turing machines, respectively. $\UL$ is a class of languages that can be accepted by an $\NL$ machine that has at most one accepting path on each input, and hence $\UL \subseteq \NL$. $\SPL$ is the class of languages whose characteristic function can be written as a logspace computable integer determinant. 
%
%
\section{Weight function}
\label{sec:wtfn}
In order to construct the desired weight function for a given graph $G_0 \in \langle \mathcal{G}_{P, W}\rangle_3$, we modify the component tree $T_0$ of $G_0$ such that it has the following properties.
\begin{itemize}
\item No two separating sets share a common vertex. \item A separating set is shared by at most two components.
\item any virtual triangle, i.e., the triangle consists of virtual edge, in a planar component is always a face.
\end{itemize}

Let $T$ be this modified component tree, and $G$ be the graph represented by $T$. We show that if we have a weight function that gives nonzero circulation to every cycle in $G$, then we can obtain a weight function that will give nonzero circulation to all the cycles in $G_0$.
Arora et al. \cite{AGGT16} showed how a component tree satisfying these properties can be obtained for $K_{3,3}$-free and $K_5$-free graphs. We give a similar construction below and show that we can modify the components of $T_0$ such that $T$ satisfies the above properties (see Section \ref{sec:modt}). Note that if the graphs inside two nodes of $T_0$ share a separating set $\tau$ and they both are constant tree-width graphs, then we can take the clique-sum of these two graphs on the vertices of $\tau$, and the resulting graph will also be a constant tree-width graph. Therefore, we can assume that if two components share a separating set, then either both of them are planar, or one of them is planar and the other is of constant tree-width.

\subsection{Modifying the Component Tree}
\label{sec:modt}
In this section, we show that how we obtain the component tree $T$ from $T_0$ so that it satisfies the above three properties.
\subparagraph{(i) No two separating sets share a common vertex:}
For a node $D$ in $T_0$, let $G_D$ be the graph inside node $D$. Assume that $G_D$ contains a vertex $v$ which is shared by separating sets $\tau_1,\tau_2, \ldots ,\tau_k$, where $k>1$, present in $G_D$. We replace the vertex $v$ with a gadget $\gamma$ defined as follows: $\gamma $ is a star graph such that $v$ is the center node and $v_1,,v_2, \ldots ,v_k$ are the leaf nodes of $\gamma$. The edges which were incident on $v$ and had their other endpoints in $\tau_i$, will now incident on $v_i$ for all $i \in[k]$. All the other edges which were incident on $v$ will continue to be incident on $v$. We do this for each vertex which is shared by more than one separating set in $G_{D}$. Let $G_{D'}$ be the graph obtained after replacing each such vertex with gadget $\gamma$. It is easy to see that if $G_D$ was a planar component, then $G_{D'}$ will also be a planar component. We show that the same holds for constant tree-width components as well.

\begin{figure}
\begin{center}
\includegraphics[scale=1]{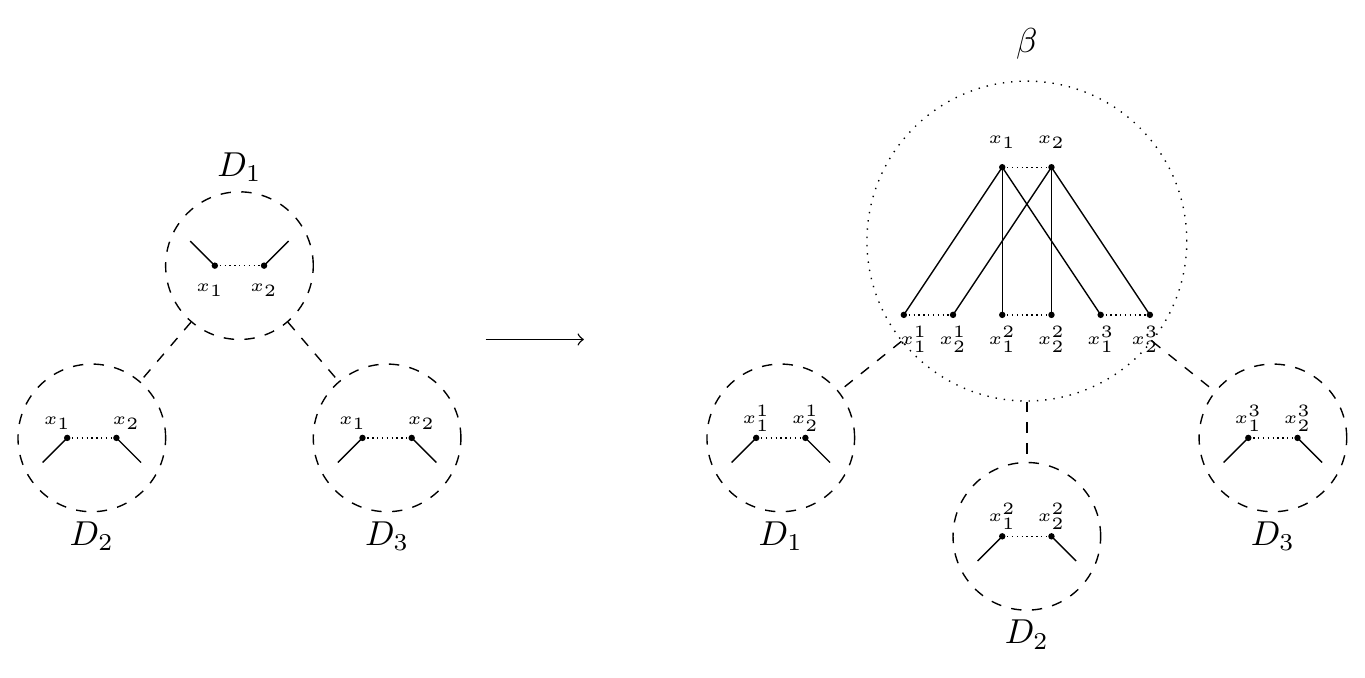}
\caption{(Left)A separating set $\{x_1,x_2\}$ is shared by components $D_1,D_2$ and $D_3$. (Right) Replace them by adding the gadget $\beta$ and connect $D_1,D_2$ and $D_3$ to $\beta$.}
\label{fig:sepset}
\end{center}
\end{figure}

\begin{claim}
If $G_D$ is a constant treewidth graph, then $G_{D'}$ will also be of constant treewidth.
\end{claim}
\begin{proof}
Let $T_D$ be a tree decomposition of $G_D$ such that each bag of $T_D$ is of constant size, i.e., contains some constant number of vertices. Let $v$ be a vertex shared by $k$ separating sets $\{x_i,y_i,v\}$, for all $i \in [k]$ in $G_D$. Let $B_1,B_2, \ldots B_k$ be the bags in $T_D$ that contain separating sets $\{x_1,y_1,v\}, \{x_2,y_2,v\} , \ldots ,\{x_k,y_k,v\}$ respectively (note that one bag may contain many separating sets). Now we obtain a tree decomposition $T_{D'}$ of the graph $G_{D'}$ using $T_D$ as follows: add the vertices $v_i$ in the bag $B_i$, for all $i \in [k]$. Repeat this for each vertex $v$ in $G_D$, which is shared by more than one separating set to obtain $T_{D'}$. Note that in each bag of $T_D$ we add at most one new vertex with respect to each separating set contained in the bag in order to obtain $T_{D'}$. Since each bag in $T_D$ can contain vertices of only constant many separating sets, size of each bag remain constant in $T_{D'}$. Also, $T_{D'}$ is a tree decomposition of $G_{D'}$.
\end{proof}

\subparagraph{(ii) A separating set is shared by at most two components:} Assume that a separating set of size $t$, $\tau =\{x_i\}_{i \leq t}$ is shared by $k$ components $D_1,D_2, \ldots D_k$, for $k>2$, in $T_0$. Let $\beta$ be a gadget defined as follows: the gadget consists of $t$ star graphs $\{\gamma_i\}_{i \leq t}$ such that $x_i$ is the center node of $\gamma_i$ and each $\gamma_i$ has $k$ leaf nodes $\{x_i^1,x_i^2, \ldots x_i^k\}$. There are virtual cliques present among the vertices $\{x_i^j\}_{i \leq t}$ for all $j \in [k]$ and among $\{x_i\}_{i \leq t}$ (see Figure \ref{fig:sepset}). If there is an edge present between any pair of vertices in the set $\{x_i\}_{i \leq t}$ in the original graph, then we add a real edge between respective vertices in $\beta$. $\beta$ shares the separating set $\{x_i^j\}_{i \leq t}$ with the component $D_j$ for all $j \in [k]$.

Note that in this construction, we create new components ($\beta$) while all the other components in the component tree remain unchanged. Notice that the tree-width of $\beta$ is constant (at most $5$ to be precise). We can define a tree decomposition of $\beta$ of tree-width $5$ as follows: $B_0,B_1',B_2', \ldots ,B_k'$ be the bags in the tree decomposition such that $B_0 = \{x_1,x_2,x_3\}$, $B_i'= \{x_1,x_2,x_3,x^i_1,x^i_2,x^i_3\}$ and there is an edge from $B_0$ to $B_i'$ for all $i \in [k]$.

\subparagraph{(iii) Any virtual triangle, i.e., the triangle consists of virtual edges, in a planar component is always a face:} $3$-cliques in a $3$-clique sum of a planar and a bounded tree-width component is always a face in the planar component. This is because suppose there is a planar component $G_i$ in which the $3$-clique on $u,v,w$ occurs but does not form a face. Then the triangle $u,v,w$ is a separating set in $G_i$, which separates the vertices in its interior $V_1$ from the vertices in its exterior $V_2$. Notice that neither of $V_1,V_2$ is empty by assumption since $u,v,w$ is not a face. However, then we can decompose $G_i$ further.

\subsection{Preserving nonzero circulation:}
\label{subsec:presnzc}
We can show that if we replace a vertex with the gadget $\gamma$, then the nonzero-circulation in the graph remains preserved: let $G_1(V_1,E_1)$ be a graph such that a vertex $v$ in $G_1$ is replaced with the gadget $\gamma$ (star graph). Let this new graph be $G_2(V_2,E_2)$. We show that if we have a skew-symmetric weight function $w_2$ that gives nonzero circulation to every cycle in $G_2(V_2,\dvec{E}_2)$, then we can obtain a skew-symmetric weight function $w_1$ that gives nonzero circulation to every cycle in $G_1(V_1,\dvec{E}_2)$ as follow. Let $u_1,u_2, \ldots, u_k$ be the neighbors of $v$ in $G_1$. For the sake of simplicity, assume that $v$ is replaced with $\gamma$ such that $\gamma$ has only two leaves $v_1$ and $v_2$ and $v$ is the center of $\gamma$. Now assume that $u_1,u_2, \ldots ,u_j$ become neighbors of $v_1$ and, $u_{j+1},u_{j+2}, \ldots , u_k$ become neighbors of $v_2$ in $G_2$, for some $j<k$. We define a function $P$ that maps each edge of $G_1$ to at most two edges of $G_2$ as follows. For edge $(u_i,v)$ in $G_1$,
\[
P(u_i,v)=
\begin{cases}
\{(u_i,v_1),(v_1,v) \},& \text{ if $u_i$ is a neighbor of $v_1$ in $G_2$, }\\
\{ (u_i,v_2),(v_2,v) \},& \text{ if $u_i$ is a neighbor of $v_2$ in $G_2$. }
\end{cases}
\]
For all the other edges $e$ of $G_1$, $P(e) = \{e \in G_2\}$. Now given weight function $w_2$ for $G_2$, we define weight function $w_1$ for $G_1$ as follows:
\begin{align*}
w_1(e) = \sum_{e' \in P(e)} w_2(e')
\end{align*}
For any ${F} \subseteq \dvec{E_1}$, we define $P({F}) = ( P(e) \mid e \in {F})$. Let $C$ be a simple cycle in $G_1$. Notice that the set of edges in $P(C)$ form a walk in $G_2$. Also, note that for some edges $e$ of $G_2$ both $e$ and $e^r$ may appear in $P(C)$. Let $\hat{E}_2(C)$ be the set of edges in $P(C)$ such that for each $e \in \hat{E}_2(C)$, both $e$ and $e^r$ appear in $P(C)$. Since our weight function is skew-symmetric, we know that $\sum_{e \in \hat{E}_2(C)}w_2(e) = 0$. Also, notice that set of edges in the set $P(C) - \hat{E}_2(C)$ form a simple cycle in $G_2$ (proof of this is same as the proof of Claim \ref{clm:simp}, that we prove later in this paper). Let $C'$ be the simple cycle in $G_2$ formed by the edges in the set $P(C) - \hat{E}_2(C)$. We know that,  $\sum_{e \in C}w_1(e) =  \sum_{e \in P(C)}w_2(e) = \sum_{e \in C'}w_2(e) + \sum_{e \in \hat{E}_2(C)}w_2(e) $ and since $\sum_{e \in \hat{E}_2(C)}w_2(e) = 0$ we have,  $\sum_{e \in C}w_1(e) = \sum_{e \in C'}w_2(e).$
%

Therefore, we can say that if $w_2$ gives nonzero circulation to $C'$, then $w_1$ gives nonzero circulation to $C$. To satisfy property $1$ on the component tree, we replace vertices of the graph with $\gamma$. Furthermore, to satisfy property $2$, we replace vertices with the gadget $\beta$, which contains nothing but multiple copies of $\gamma$. Thus from above, we can conclude that these constructions preserve the nonzero circulation. Now we will work with the graph $G$ and the component tree $T$.
\subsection{Tree decomposition}
Note that the component tree $T$ of $G$ is also a tree decomposition of $G$ in the sense that we can consider the nodes of the component tree as bags of vertices. We know that $T$ contains two types of nodes: (i) nodes that contain planar graphs, (ii) nodes that contain constant tree-width graphs. We call them \pty and \cty nodes, respectively.

Now we will construct another tree decomposition $T'$ of $G$ using the component tree $T$. $T'$ have two types of bags: (i) A bag with respect to each \pty node of $T$, which contains the same set of vertices as the \pty node, and (ii) bags obtained from tree decomposition of the component inside each \cty node. For a node $N$ of $T$, let $G_N$ denote the graph inside node $N$. Let $V(G_N)$ denote the vertices in the graph $G_N$ and $T_N$ be a tree decomposition of $G_N$ obtained using Lemma \ref{lem:boundtw}.

\begin{itemize}

\item Bags of $T'$ are defined as follows: $\{V({G_N}) \mid N \in T, \text{ where $N$ is a \pty node in $T$} \}\; \bigcup \; \{B \mid B \in T_{N}, \text{ where $N$ is a \cty node in $T$}\}$.
We know that in a tree decomposition of a graph $H$, for each clique in $H$, there exists a bag in the tree decomposition of $H$ that contains all the vertices of the clique. Let $\tau$ be a separating set present in the graph contained in a \cty node $N$ of $T$. While constructing a tree decomposition $T_N$, we consider the virtual clique present among the vertices of $\tau$ as a part of $G_{N}$. This ensures that there exists a bag in $T_N$ that contains all the vertices of $\tau$.

\item Edges in $T'$ are defined as follows: (i) for a $\cty$ node $N$, let $B$ and $B'$ be two bags in $T_N$. If an edge in $T_N$ connects $B$ and $B'$, then add an edge between them in $T'$ as well,(ii) let $N'$ and $N''$ be two adjacent nodes in $T$ such that they share a separating set $\tau$. As mentioned above, we know that either both $N'$ and $N''$ are \pty or one of them is \pty and the other one is a \cty node. If both of them are \pty: we add an edge between the bags in $T'$ that contain the vertices $V(G_{N'})$ and $V(G_{N''})$. If $N'$ is a \pty and $N''$ is a \cty node: remember that we replaced node $N''$ by its tree decomposition. let $B$ be any bag in $T_{N''}$ which contains all the vertices of $\tau$. We add an edge between the bag containing vertices $V(G_N)$ and $B$.

\end{itemize}

It is easy to see that $T'$ is also a tree decomposition of $G$. From Lemma \ref{lem:boundtw}, we can say that the overall construction of $T'$ remains in logspace. For simplicity, we rename the bags of $T'$ to $B_1,B_2, \ldots ,B_k$. Let $B_i$ be a bag in $T'$ corresponding to a node $N$ in $T$. If $N$ contains a separating set $\tau$, then the set of vertices of $\tau$ is also called a separating set in $B_i$. We will need this notion later on while constructing the weight function.

Note that Arora et al. \cite{AGGT16} obtained a component tree $\hat{T}$ of an input graph $\hat{G}$ such that each node of $\hat{T}$ contains either a planar graph or a constant size graph. Then they show that for a cycle $C$ in $\hat{G}$, the nodes which contain edges of $C$ form a connected component of $\hat{T}$. They use this property to construct the desired weight function. Here we can say that $T'$ is also a component tree of $G$ such that each node of $T'$ contains a planar graph or a constant size graph, in a sense that we assign each edge $e$ of $G$ to one of the bags of $T'$, which contains both the endpoints of $e$. However, we cannot claim that for a cycle $C$, the bags containing edges of $C$ form a connected component in $T'$ (for example, see Figure \ref{fig:btw}). Thus in this paper, we use the tree decomposition $T'$ to construct another graph $G'$ and associate edges of $G'$ to the bags to $T'$ such that for each cycle $C'$ in $G'$, the bags which have edges of $C'$ associated with them, form a connected component in $T'$. We show that if we can construct a skew-symmetric weight function for $G'$ such that it gives nonzero circulation to every cycle in $G'$, then we can obtain a skew-symmetric weight function for $G$, which gives nonzero circulations to all the cycles in $G$.

\begin{figure}
\begin{center}
\includegraphics[scale=2]{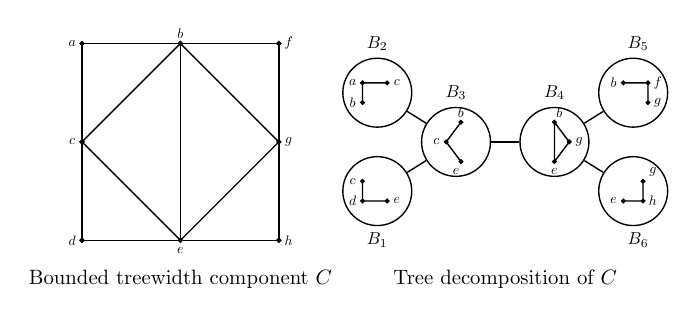}
\caption{Assume $C$ (left) is a constant treewidth component of $G$. We replace $C$ with its tree decomposition (right) in the component tree $T'$ and associated each edge of $C$ with one of the bags in the tree decomposition, as shown in the figure. Notice that for cycle \textit{abfghedca} the bags ($B_1,B_2,B_5$ and $B_6$), which have edges of the cycle associated with them, do not form a connected component.}
\label{fig:btw}
\end{center}
\end{figure}

\subsection{Construction of G'}
We construct $G'(V',E')$ from the tree decomposition $T'$ of $G(V,E)$ as follows. We borrow notation from Datta et al. \cite{DKMTVZ20}. Without loss of generality, assume that $T'$ is a rooted tree and the parent-child relationship is well defined.
\begin{itemize}
\item Vertex set $V' = \{v_{B_i} \mid B_i \in T' , v \in B_i\}$, i.e., for each vertex $v$ of the $G$, we have copies of $v$ in $G'$ for each bag $B_i$ of $T'$ in which $v$ appears. Copies of vertices of a separating set are also called a separating set in $G'$.
\item Edge set $E' = \{(u_{B_i},v_{B_i}) \mid (u,v) \in E, \; u \notin parent(B_i) \; or \; v \notin parent(B_i)\} \; \bigcup \; \{(v_{B_i},v_{B_j}) \mid \text{$B_i$ and $B_j$ are adjacent in $T'$}\}$. In other words, we add an edge between the two vertices $u$ and $v$ of same bag $B_i$ if there is an edge between those vertices in $G$ and no ancestor of $B_i$ in $T'$ contains both the vertices $u$ and $v$. We add an edge between two copies of a vertex if the bags they belong to, are adjacent in $T'$.
\end{itemize}

\begin{lemma}
\label{lem:wtforg}
Given a polynomially bounded, skew-symmetric weight function $w'$ that gives a nonzero circulation to every cycle in $G'$, we can find a polynomially bounded, skew-symmetric weight function $w$ for $G$ that gives a nonzero circulation to every cycle in $G$.
\end{lemma}
Proof of Lemma \ref{lem:wtforg} is given in the Appendix. Now the only thing remaining is the logspace construction of the polynomially bounded skew-symmetric weight function $w'$.

\subsection*{Constructing Weight Function for G'}\label{subsec:wtandaux} To construct the weight function $w'$ for $G'$, we associate each edge of $G'$ with some bag of $T'$. Let $(u_{B_i},v_{B_j})$ be an edge in $G'$: (i) if $i=j$, i.e., if $B_i$ and $B_j$ are the same bags. In this case associated $(u_{B_i},v_{B_j})$ with that bag, (ii) if $i \neq j$: by our construction of $G'$ we know that either $B_i$ is the parent of $B_j$ or $B_j$ is the parent of $B_i$ (i.e. $u_{B_i}$ and $v_{B_j}$ are the copies of a same vertex of $G$). In both the cases, associate $(u_{B_i},v_{B_j})$ with the parent bag. We will use the following claim later in the paper.

\begin{claim}
\label{clm:connect}
For any cycle, $C$ in $G'$, the bags of $T'$ which have some edge of $C$ associated with them, form a connected component in $T'$.
\end{claim}

The weight function $w'$ is similar to the one constructed for $K_{3,3}$-free and $K_5$ free graphs \cite{AGGT16}. We assign weights to the edges of the graph $G'$ depending upon the height of the bag they are associated with. The weights assigned to them are exponential in the height of the bags. Therefore, we need the height of a bag to be $O(\log n)$ to obtain a polynomially bounded weight function. Hence similar to \cite{AGGT16}, we define an auxiliary tree $A(T')$ of the tree $T'$. In some sense, $A(T')$ is a balanced representation of $T'$, therefore the height of $A(T')$ is $O(\log n)$. Nodes in $A(T')$ are the same as $T'$, i.e. the bags of $T'$, but the edges between the bags are inserted differently. The weight of an edge of $G'$ associated with a bag $B$ of $T'$ depends upon the height of the bag $B$ in $A(T')$.

\subparagraph*{Auxiliary Tree:} In order to construct the auxiliary tree from $T'$, first, we find a node called \textit{center} node $c(T')$. Make this node the root of the $A(T')$. Let $T_1,T_2, \ldots ,T_l$ be the subtrees obtained by deleting $c(T')$ from $T'$. Recursively apply the same procedure on these subtrees and make $c(T_1),c(T_2), \ldots, c(T_l)$ children of $c(T)$. If $c(T)$ shares a separating set $\tau$ with $T_i$ then $c(T_i)$ is said to be attached at $\tau$ with $c(T)$. Center nodes are chosen in such a way that the resultant tree $A(T')$ has height $O(\log n)$. Readers are referred to Section 3.3 of Arora et al. \cite{AGGT16} to see the logspace construction of $A(T')$. We use the same construction here. The height of the root of $A(T')$ is defined as the number of nodes in the longest path from the root $c(T)$ to a leaf node. The heights of other nodes are one less than that of their parent. From now on, we work with $A(T')$. We use the following two properties of the auxiliary tree.
\begin{enumerate}[(i)]
\item Height of a node in $A(T')$ is $O(\log n)$.
\item If ${\tilde{T}}$ is a subtree of $T'$, then there exists a bag $B$ in $\tilde{T}$ such that all the other bags of $\tilde{T}$ are descendants of $B$ in $A(T')$.
\end{enumerate}
Now we define the weight function $w'$ for the graph $G'$. Note that the graph induced by the set of edges associated with a bag $B_i$ is either planar or constant size; we call these the components of $G'$. $w'$ is a linear combination of two weight functions $w_1$ and $w_2$. $w_1$ gives nonzero circulation to those cycles which are completely contained within a component, and $w_2$ gives nonzero circulation to those cycles which span over at least two components. We define $w_1$ and $w_2$ separately for planar and constant size components. Let $G'_{B_i}$ be the graph induced by the set of edges associated with the bag $B_i$. Let $K$ be a constant such that $K > max(2^{m+2},7)$, where $m$ is the maximum number of edges associated with any constant size component.

\subparagraph{ If $G'_{B_i}$ is a planar component:}
$w_1$ for such components is same as the weight function defined in \cite{BTV09} for planar graphs. We know that given a planar graph $G$, its planar embedding can be computed in logspace \cite{AllenderMahajan04}.

\begin{theorem}[\cite{BTV09}]
Given a planar embedding of a graph $H$, there exists a logspace computable function $w$ such that for every cycle $C$ of $H$, circulation of the cycle $w(C) \neq 0$.
\end{theorem}
The above weight function gives nonzero circulation to every cycle that is completely contained in a planar component.

The weight function $w_2$ for planar components is defined as follows. $w_2$ assigns weights to only those faces of the component, which are adjacent to some separating set. For a subtree of $T_s$ of $A(T')$, let $l(T_s)$ and $r(T_s)$ denote the number of leaf nodes in $T_s$ and root node of $T_s$, respectively. For a bag $B_i$, $h(B_i)$ denotes the height of the bag in $A(T')$. If $B_i$ is the only bag in the subtree rooted at $B_i$, then each face in $G'_{B_i}$ is assigned weight zero. Otherwise, let $\tau$ be a separating set where some subtree $T_i$ is attached to $B_i$. The faces adjacent to $\tau$ in $G'_{B_i}$ are assigned weight $2\times K^{h(r(T_i))} \times l(T_i)$. If a face is adjacent to more than one separating set, then the weight assigned to the face is the sum of the weights due to each separating set. The weight of a face is defined as the sum of the weights of the edges of the face in clockwise order. If we have a skew-symmetric weight function, then the weight of the clockwise cycle will be the sum of the weights of the faces inside the cycle \cite{BTV09}. Therefore assigning positive weights to every face inside a cycle will ensure that the circulation of the cycle is nonzero. Given weights on the faces of a graph, we can obtain weights for the edges so that the sum of the weights of the edges of a face remains the same as the weight of the face assigned earlier \cite{Kor09}.

\subparagraph{If $G'_{B_i}$ is a constant size component:} For this type of component, we need only one weight function. Thus we set $w_2$ to be zero for all the edges in $G'_{B_i}$ and $w_1$ is defined as follows. Let $e_1,e_2, \ldots ,e_k$ be the edges in the component $Q_i$, for some $k \leq m$. Edge $e_j$ is assigned weight $2^i \times K^{h(r(T_i))-1} \times l(T_{i})$ (for some arbitrarily fixed orientation), Where $T_{i}$ is the subtree of $A(T')$ rooted at $B_i$. Note that for any subset of edges of $G'_{B_i}$, the sum of the weight of the edges in that subset is nonzero with respect to $w_1$.

The final weight function is $w' = \langle w_1 + w_2 \rangle.$ Since the maximum height of a bag in $A(T')$ is $O(\log n)$, the weight of an edge is at most $O( n^c)$, for some constant $c>0$.

\begin{lemma}
\label{lem:maxval}
For a cycle $C$ in $G'$ sum of the weights of the edges of $C$ associated with the bags in a subtree $T_i$ of $A(T')$ is $< K^{h(r(T_i))} \times l(T_i) $.
\end{lemma}

\begin{proof}
Let $w(C_{T_i})$ denotes the sum of the weight of the edges of a cycle $C$ associated with the bags in $T_i$. We prove the Lemma by induction on the height of the root of the subtrees of $A(T')$. Note that the Lemma holds trivially for the base case when the height of the root of a subtree is $1$.

\textit{Induction hypothesis}: Assume that it holds for all the subtrees such that the height of their root is $<h(r(T_i))$.

Now we will prove it for $T_i$. Let $T_i^1,T_i^2, \ldots , T_i^k$ be the subtrees attached to $r(T_i)$.
\begin{itemize}
\item First, consider the case when $G'_{r(T_i)}$ is a constant size graph: In this case, we know that the sum of the weights of the edges of $C$ associated with $r(T_i)$ is $ \leq \sum_{j=1}^m 2^j \times K^{h(r(T_i))-1} \times l(T_i) $ and by the induction hypothesis, we know that $w(C_{T_i^j}) < K^{h(r(T_i^j))} \times l(T_i^j)$, for all $j\in [k]$. Therefore,
\begin{eqnarray*}
w(C_{T_i}) &\leq& \sum_{j=1}^m 2^j \times K^{h(r(T_i))-1} \times l(T_i) + \sum_{j=1}^k K^{h(r(T_i^j))} \times l(T_i^j) \\
w(C_{T_i}) &\leq& (2^{m+1}-1) \times K^{h(r(T_i))-1} \times l(T_i) + K^{h(r(T_i))-1} \times l(T_i) \\
w(C_{T_i}) &\leq & (2^{m+1}) \times (K^{h(r(T_i))-1} \times l(T_i)) \hspace{4cm} [ K > 2^{m+2}] \\
w(C_{T_i}) &<& K^{h(r(T_i))} \times l(T_i)
\end{eqnarray*}

\item When $G'_{r(T_i)}$ is a planar graph: let $\tau_1,\tau_2, \ldots ,\tau_k$ be the separating sets present in $G'_{r(T_i)}$ such that the subtree $T_i^j$ is attached to $r(T_i)$ at $\tau_j$, for all $j \in [k]$. A separating set can be present in at most 3 faces. Thus it can contribute $2 \times 3 \times K^{h(r(T_i^j))} \times l(T_i^j)$ to the circulation of the cycle $C$. Therefore,
\begin{eqnarray*}
w(C_{T_i}) &\leq & \sum_{j=1}^k 6 \times K^{h(r(T_i^j))} \times l(T_i^j) + \sum_{j=1}^k K^{h(r(T_i^j))} \times l(T_i^j) \\
w(C_{T_i}) &\leq& 7 \times K^{h(r(T_i))-1} \times l(T_i)\hspace{5cm} [K > 7] \\
w(C_{T_i}) &<& K^{h(r(T_i))} \times l(T_i)
\end{eqnarray*}
\end{itemize}
\end{proof}

\begin{lemma}
\label{lem:domwt}
For a cycle, $C$ in $G$ let $B_i$ be the unique highest bag in $A(T')$ that have some edges of $C$ associated with it. Then the sum of the weights of the edges of $C$ associated with $B_i$ will be more than that of the rest of the edges of $C$ associated with the other bags.
\end{lemma}

\begin{proof}
Let $T_i$ be the subtree of $A(T')$ rooted at $B_i$. We know that sum of the weights of the edges of $C$ associated $B_i$ is $\geq 2 \times K^{h(r(T_i))-1} \times l(T_i)$. Let $T_i^1, T_i^2 , \ldots ,T_i^k$ be the subtree of $T_i$ rooted at children of $B_i$. By Lemma \ref{lem:maxval}, we know that the sum of the weight of the edges of $C$ associated with the bags in these subtrees is $< \sum_{j=1}^k K^{h(r(T_i^j))} \times l(T_i^j) = K^{h(r(T_i))-1} \times l(T_i)$. Therefore, the lemma follows.
\end{proof}

\begin{lemma}
\label{lem:nonzerowt}
Circulation of a simple cycle $C$ in the graph $G'$ is nonzero with respect to $w'$.
\end{lemma}

\begin{proof}
If $C$ is contained within a component, i.e., its edges are associated with a single bag $B_i$, then we know that $w_1$ assigns nonzero circulation to $C$. Suppose the edges of $C$ are associated with more than one bag in $T'$. By Claim \ref{clm:connect}, we know that these bags form a connected component. By the (ii) property of $A(T')$, we know that there is a unique highest bag $B_i$ in $A(T')$ which have edges of $C$ associated with it. Therefore from Lemma \ref{lem:domwt} we know that the circulation of $C$ will be nonzero.
\end{proof}

\begin{proof}[Proof of Theorem \ref{thm:main}]
Proof of Theorem \ref{thm:main} follows from Lemma \ref{lem:wtforg} and \ref{lem:nonzerowt}.
\end{proof}
\section{Maximum Matching}\label{sec:staticmatch}
In this section, we consider the complexity of the maximum matching problem in single crossing minor free graphs. Recently Datta et al.~\cite{DKKM18} have shown that the bipartite maximum matching can be solved in \SPL\ in the planar, bounded genus, $K_{3,3}$-free and $K_{5}$-free graphs.

Their techniques can be extended to any graph class where nonzero circulation weights can be assigned in logspace. For constructing a maximum matching in $K_{3,3}$-free and $K_{5}$-free bipartite graphs, they use the logspace algorithm of~\cite{AGGT16} as a black box. Since from Theorem~\ref{thm:main} nonzero circulation weights can be computed for the more general class of any single crossing minor free graphs, we get the bipartite maximum matching result of Corollary~\ref{cor:stat}.

In a related work recently, Eppstein and Vazirani~\cite{EV21} have shown an $\NC$ algorithm for the case when the graph is not necessarily bipartite. However, the result holds only for constructing perfect matchings. In non-bipartite graphs, there is no known parallel (e.g., \NC) or space-efficient algorithm for deterministically constructing a maximum matching even in the case of planar graphs~\cite{Sankowski18, DKKM18}. Datta et al.~\cite{DKKM18} givelo an approach to design a \emph{pseudo-deterministic} \NC\ algorithm for this problem. Pseudo-deterministic algorithms are probabilistic algorithms for search problems that produce a unique output for each given input with high probability. That is, they return the same output for all but a few of the possible random choices. We call an algorithm pseudo-deterministic $\NC$ if it runs in $\RNC$ and is pseudo-deterministic.

Using the Gallai-Edmonds decomposition theorem, \cite{DKKM18} shows that the search version of the maximum matching problem reduces to determining the size of the maximum matching in the presence of algorithms to (a) find a perfect matching and to (b) solve the bipartite version of the maximum matching, all in the same class of graphs. This reduction implies a pseudo-deterministic $\NC$ algorithm as we only need to use randomization for determining the size of the matching, which always returns the same result. For single crossing minor free graphs, using the $\NC$ algorithm of \cite{EV21} for finding a perfect matching and our $\SPL$ algorithm for finding a maximum matching in bipartite graphs, we have the following result:
\begin{theorem}
Maximum matching in single-crossing minor free graphs (not necessarily bipartite) is in pseudo-deterministic \NC.
\end{theorem}

\section{Conclusion}
\label{sec:concl}
We have given a construction of a nonzero circulation weight function for the class of graphs that can be expressed as 3-clique-sums of planar and constant treewidth graphs. However, it seems that our technique can be extended to the class of graphs that can be expressed as 3-clique-sums of constant genus and constant treewidth graphs. Further extending our results to larger graph classes would require fundamentally new techniques. This is so because the most significant bottleneck in parallelizing matching algorithms for larger graph classes such as apex minor free graphs or $H$-minor free graphs for a finite $H$ is the absence of a parallel algorithm for the structural decomposition of such families. Thus we would need to
revisit the Robertson-Seymour graph minor theory to parallelize it. This paper thus serves the dual purpose of delineating the boundaries of the known regions of parallel (bipartite) matching and reachability and as an invitation to the vast unknown of parallelizing the Robertson-Seymour structure theorems.
\bibliography{references}
\appendix
\section{Appendix}
\label{sec:appendix}
\begin{proof}[Proof of Lemma \ref{lem:wtforg}]
To construct the weight function $w$, we associate a sequence $P(u,v)$ of edges of $G'$ with each edge $(u,v)$ of $G$. Assume that $T'$ is a rooted tree, and root is the highest node in the tree. The heights of all the other nodes are one less than that of their parent. As we mentioned that $T'$ is a tree decomposition of $G$. For an edge $(u,v)$, we know that there are unique highest bags $B_1$ and $B_2$ that contain vertices $u$ and $v$, respectively.
\begin{itemize}
\item If $B_1 = B_2$ then $P(u,v) = (u_{B_1} , v_{B_2})$.

\item If $B_1$ is an ancestor of $B_2$ then \\$P(u,v) = (u_{B_1}, u_{parent(\cdots (parent(B_2)))}), \ldots ,(u_{parent(B_2)}, u_{B_2}),(u_{B_2},v_{B_2})$.

\item If $B_1$ is a descendant of $B_2$ then \\ $P(u,v) = (u_{B_1}, v_{B_1}), (v_{B_1}, v_{parent(B_1)}), \ldots , (v_{parent(\cdots (parent(B_1)))},v_{B_2})$.
\end{itemize}

The weight function $w$ for the graph $G$ is defined as follows:
\begin{eqnarray*}
w(u,v) = \sum_{e\in P(u,v)} w'(e)
\end{eqnarray*}

For a simple cycle $C = e_1,e_2, \ldots ,e_j$ in $G$,  we define $P(C) = P(e_1),P(e_2), \ldots ,P(e_j)$. Note that $P(C)$ is a closed walk in $G'$. Let $E'_d(C)$ be the subset of edges of $G'$ such for all edges $e \in E'_d(C)$ both $e$ and $e^{r}$ appear in $P(C)$, where $e^r$ denotes the edge obtained by reversing the direction of $e$. We prove that if we remove the edges of $E'_d(C)$ from $P(C)$ then the remaining edges $P(C) - E'_d(C)$  form a simple cycle in $G'$. 

\begin{claim}
\label{clm:simp}
Edges in the set $P(C)- E'_d(C)$ form a simple cycle in $G'$. 
\end{claim}

\begin{proof}
Note that the lemma follows trivially if $P(C)$ is a simple cycle. Therefore, assume that $P(C)$ is not a simple cycle. We start traversing the walk $P(C)$ starting from the edges of the sequence $P(e_1)$. Let $P(e_k)$ be the first place where a vertex in the walk $ P(e_1) P(e_2).....P(e_k)$ repeats, i.e., edges in the sequence  $P(e_1) P(e_2).....P(e_{k-1})$ form a simple path, but after adding the edges of  $P(e_k)$ some vertices are visited twice in the walk $P(e_1) P(e_2).....P(e_{k-1})P(e_k)$ for some $k \leq j$. This implies that some vertices are visited twice in the sequence $P(e_{k-1})P(e_k)$. Let $e_{k-1}=(u,v)$ and $e_k=(v,x)$. This implies that some copies of the vertex $v$ appear twice in the sequence $P(u,v)P(v,x)$. Let $B_1$ and $B_2$ be the highest bags such that $B_1$ contains the copies of vertices $u$ and $v$, and $B_2$ contains the copies of $v$ and $x$. Let bag $B$ be the lowest common ancestor of $B_1$ and $B_2$. We know that $B$ must contain a copy of the vertex $v$, i.e., $v_B$. Let $B'$ be the highest bag containing  a copy of vertex $v$, i.e., $v_{B'}$. First, consider the case when neither $B_1$ is an ancestor of $B_2$ and vice-versa, other cases can be handled similarly. In that case sequence  $P(u,v) = u_{B_1}v_{B_1}  v_{parent(B_1)} \ldots v_{B} \ldots v_{B'}$ and $P(v,w) = v_{B'} \ldots v_{B} \ldots v_{parent(B_2)} v_{B_2}w_{B_2}$. Note that in $P(u,v)$ a path goes from $v_B$ to $v_{B'}$ and the same path appear in reverse order from $v_{B'}$ to $v_{B}$ in the sequence $P(v,x)$. Therefore if we remove these two paths from $P(u,v)$ and $P(v,w)$ the remaining subsequence of the sequence $P(u,v)P(v,w)$ will be a simple path, i.e., no vertex will appear twice since $B$ is the lowest common ancestor of $B_1$ and $B_2$. Now repeat this procedure for $P_{k+1},P_{k_2} \ldots$ and so on till $P_{k}$. In the end, we will obtain a simple cycle.
\end{proof}

Since we assumed that the weight function $w'$ is skew-symmetric, we know that $w'(e) = -w (e^{r})$, for all $e \in G'$. This implies that $w'(E'_d(C)) = 0 $. Therefore $w(C) = w'(P(C)) = w'(P(C) - E'_d(C))$. From Claim \ref{clm:simp} we know that edges in the set $P(C) - E'_d(C)$ form a simple cycle and we assumed that $w'$ gives nonzero circulation to every simple cycle therefore, $w'(P(C) - E'_d(C)) \neq 0$. This implies that $w(C) \neq 0$. This finishes the proof of Lemma \ref{lem:wtforg}.
\end{proof}

\begin{proof}[Proof of Claim \ref{clm:connect}]
Note that if we treat each vertex in the bags of $T'$ distinctly, then there is a one-to-one correspondence between vertices of $G'$ and vertices in the bags of $T'$. Therefore, in $T'$ a vertex of $G'$ is identified by its corresponding vertex. Note that all the bags which contain vertices of the cycle $C$ form a connected component in $T'$. We will now prove that if a bag $B$ contains some vertices of $C$, then either $B$ has some edges of $C$ associate with it or no bag in the subtree rooted at $B$ has any edge of $C$ associated with it. From this, we can conclude that the bags which have some edges of $C$ associated with them form a connected component in $T'$.

Assume that $B$ is a bag which contains a vertex of $C$ but no edge of $C$ is associated with it. This implies that $C$  never enters in any of the children of $B$. Because, let us assume it enters to some child $B'$ of $B$ through some vertex $v_{B'}$ of $B'$. In that case, there will be an edge $(v_{B},v_{B'})$ of $C$ associated with the bag $B$, which is a contradiction. Therefore subtree rooted at $B$ will not have any edge of the cycle $C$ associated with it. This finishes the proof.
\end{proof}

\end{document}